\documentclass[11pt]{article}

\usepackage[margin=1.0in]{geometry}
\usepackage{amsmath}
\usepackage[noend]{algpseudocode}
\usepackage{algorithm}
\usepackage{amsthm}
\usepackage{amsfonts}
\usepackage{graphicx}
\usepackage{hyperref}
\usepackage{color}

\newtheorem{theorem} {Theorem}
\newtheorem{lemma} {Lemma}
\newtheorem{definition} {Definition}
\newtheorem{corollary} {Corollary}

\newcommand{\ceil}[1]{\left \lceil #1 \right \rceil}
\newcommand{\E}{{\mathbb{E}}}
\newcommand{\eps}{\varepsilon}
\newcommand{\poly}{\operatorname{poly}}

\newcommand{\var}[1]{\operatorname{#1}}
\newcommand{\func}[1]{\operatorname{\textsc{#1}}}

\newcommand{\rank}{R}

\title{Optimal Quantile Approximation in Streams}

\author{
Zohar Karnin\\ \texttt{Yahoo Research}\\ \texttt{\small zkarnin@yahoo-inc.com}
\and 
Kevin Lang\\ \texttt{Yahoo Research}\\ \texttt{\small langk@yahoo-inc.com}
\and
Edo Liberty\\ \texttt{Yahoo Research}\\ \texttt{\small edo@yahoo-inc.com}
}
\date\nonumber
\begin{document}
\maketitle

\begin{abstract}
This paper resolves one of the longest standing basic problems in the streaming computational model. 
Namely, optimal construction of quantile sketches.
An $\varepsilon$ approximate quantile sketch receives a stream of items $x_1,\ldots,x_n$ and allows one to approximate the rank of any query up to additive error $\varepsilon n$ with probability at least $1-\delta$.
The rank of a query $x$ is the number of stream items such that $x_i \le x$.
The minimal sketch size required for this task is trivially at least $1/\varepsilon$.
Felber and  Ostrovsky obtain a $O((1/\varepsilon)\log(1/\varepsilon))$ space sketch for a fixed $\delta$.
To date, no better upper or lower bounds were known even for randomly permuted streams or for approximating a specific quantile, e.g.,\ the median.
This paper obtains an $O((1/\varepsilon)\log \log (1/\delta))$ space sketch and a matching lower bound. 
This resolves the open problem and proves a qualitative gap between randomized and deterministic quantile sketching.
One of our contributions is a novel representation and modification of the widely used merge-and-reduce construction. 
This subtle modification allows for an analysis which is both tight and extremely simple.
Similar techniques should be useful for improving other sketching objectives and geometric coreset constructions.

\end{abstract}

\section{Introduction}
Given a set of items $x_1,\ldots,x_n$, the quantile of a value $x$ is the fraction of items in the stream such that $x_i \le x$.
It is convenient to define the rank of $x$, $\rank(x)$, as the \emph{number} of items such that $x_i \le x$.
An additive error $\eps n$ for $\rank(x)$ is an $\eps$ approximation of its rank.
The literature distinguishes between several different definitions of this problem.
In this manuscript we distinguish between the single quantile approximation problem and the all quantiles approximation problem.
\begin{definition} The single quantile approximation problem:
Given $x_1,\ldots,x_n$ in a streaming fashion in arbitrary order, construct a data structure for 
computing $\tilde{\rank}(x)$.
By the end of the stream, receive a single element $x$ and compute $\tilde{\rank}(x)$ such that $| \tilde{R}(x) - \rank(x) | \le \eps n$ with probability $1-\delta$.
\end{definition}
There are variations of this problem in which, the algorithm is not given $x$ (as a query) but rather a rank $r$. 
It should be able to provide an element $x_i$ from the stream such that $|\rank(x_i) - r| \le \eps n$.
There are also variants that make the value $r$, or $r/n$ known to the algorithm in advance. 
For example, one could a priori choose to search for an approximate median.\footnote{We mention that there are easy reductions between the different variants that maintain the failure probability and error up to a constant.}
The solution we propose solves all the above variants by solving a harder task called the \emph{all quantiles} problem. 
\begin{definition} The all quantiles approximation problem:
Given $x_1,\ldots,x_n$ in a streaming fashion in arbitrary order, construct a data structure for computing $\tilde{\rank}(x)$.
By the end of the stream, with probability $1-\delta$, for all values of $x$ simultaneously it should hold that $| \tilde{\rank}(x) - \rank(x) | \le \eps n$.
\end{definition}
Observe that approximating a set of $O(1/\eps)$ single queries well suffices for solving the all quantiles approximation problem.
Therefore, solving the single quantiles approximation problem with failure probability at most $\eps\delta$ constitutes a valid solution for the 
all quantiles approximation problem simply by invoking the union bound.

\subsection{Related Work}
Two recent surveys \cite{Wang13}\cite{GK2016} on this problem give ample motivation and explain the state of the art in terms of algorithms and theory in a very accessible way.\footnote{Manuscript \cite{GK2016} was authored in 2007 as a book chapter. It does not contain recent results but is an excellent survey nonetheless.} In what follows, we shortly review some of the prior work that is most relevant in the context of this manuscript.
For readability, space complexities of randomized algorithms apply for a constant success probability unless otherwise stated.

Manku, Rajagopalan and Lindsay \cite{Manku99} built on the work of Munro and Paterson \cite{MUNRO1980315} and gave a randomized solution which uses at most $O((1/\eps) \log^2(n \eps))$ space. 
A simple deterministic version of their algorithm achieves the same bounds. This was pointed out, for example, by \cite{Wang13}. 
We refer to their algorithm as MRL. 
Greenwald and Khanna \cite{GreenwaldK01} created an intricate deterministic algorithm that requires $O((1/\eps) \log(n \eps))$ space.
This is the best known deterministic algorithm for this problem.
We refer to their algorithm as GK.

Allowing randomness enables sampling.  
A uniform sample of size $n' = O(\log(1/\eps)/\eps^2)$ from the stream suffices to produce an all quantiles sketch. 
Feeding the sampled elements into a GK sketch yields an $O((1/\eps) \log(1/\eps))$ solution.
However, to produce such samples, one must know $n$ (at least approximately) in advance. 
This observation was already made by Manku et al. \cite{Manku99}.
Since $n$ is not known in advance it is not a trivial task to combine sampling with GK sketches.
Recently, Felber and Ostrovsky \cite{FelberO15a} managed to 
do exactly that. They achieved space complexity of $O((1/\eps)\log(1/\eps))$ by using sampling and several GK sketches in conjunction in a non trivial way. To the best of our knowledge, this is the best known space complexity result to date.

\paragraph{Mergeability:} An important property of sketches is called mergability \cite{Agarwal12}.
Informally, this property allows one to sketch different sections of the stream independently and then combine the resulting sketches.
The combined sketch should be as accurate as a single sketch one would have computed had the entire stream been consumed by a single sketcher. This is formally stated in Definition~\ref{def:merge}.
\begin{definition} \label{def:merge}
Let $S$ denote a sketching algorithm mapping a stream $N$ to a summary $S(N)$, and denote by $\eps(S(N),N)$ the error associated with the summary $S(N)$ on the stream $N$. 
Let $N_1,N_2$ denote two sequences of items, and let $N=[N_1, N_2]$ denote the sequence obtained by concatenating $N_1,N_2$.
A sketching algorithm $S$ is mergeable if there exists a merge operation $M$ such that
\[
\eps(M(S(N_1),S(N_2)),N) \le \eps(S(N),N)
\]
\end{definition}
This property is extremely important in practice since large datasets are often distributed across many machines.
Agarwal et al \cite{Agarwal12} conjecture that the GK sketch is not mergeable. They describe a mergeable sketch 
of space complexity $(1/\eps)\log^{3/2}(1/\eps)$. It is worth noting that this results predates that of \cite{FelberO15a}.


\paragraph{Model:} Different sketching algorithms perform different kinds of operations on the elements in the stream.
The most restricted model is the comparison model. 
In this model, there is a strong order imposed on the elements and the algorithm can only compare two items to decide which is larger.
This is the case, for example, for lexicographic ordering of strings. 
All the cited works above operate in this model.
Another model assumes the total size of the universe is bounded by $|U|$. 
An $O((1/\eps)\log(|U|))$ space algorithm was suggested by \cite{Shrivastava04} in that model.
If the items are numbers, for example, one could also compute averages or perform gradient descent like algorithms such as \cite{brodnik2013space}.
In machine learning, a quantile loss function refers to an asymmetric or weighted $\ell_1$ loss. 
Predicting the values of the items in the stream with the quantile loss converges to the correct quantiles.
Such methods however only apply to randomly shuffled streams.

\paragraph{Lower bounds:} For any algorithm, there exists a (trivial) space lower bound of $\Omega(1/\eps)$.
Hung and Ting \cite{hung2010omega} showed that any \emph{deterministic} comparison based algorithm for the  \ single quantile \  approximation \  problem must store $\Omega((1/\eps)\log(1/\eps))$ items. 
Felber and Ostrovsky \cite{FelberO15a} suggest, as an open problem, that the $\Omega ((1/\eps) \log(1/\eps))$ lower bound could potentially hold for randomized algorithms as well. 
Prior to this work, it was very reasonable to believe this conjecture is true. 
For example, no $o((1/\eps)\log(1/\eps))$ space algorithm was known even for randomly permuted streams.\footnote{If the stream is presented in random order, obtaining an $O((1/\eps)\log(1/\eps))$ solution is easy.
Namely, save the first $O((1/\eps)\log(1/\eps))$ items and, from that point on, only count how many items fall between two consecutive samples.
Due to coupon collector, at least one sample will fall in each stretch of $\eps n$ items which makes the solution trivially correct.}

\subsection{Main Contribution}
We begin by re-explaining the work of \cite{Agarwal12} and the deterministic version of \cite{Manku99} from a slightly different view point.
The basic building block for these algorithms is a single compactor.
A compactor can store $k$ items all with the same weight $w$.
It can also compact its $k$ elements into $k/2$ elements of weight $2w$ as follows.
First, the items are sorted. Then, either the even or the odd elements in the sequence are chosen.
The unchosen items are discarded.
The weight of the chosen elements is doubled (set to $2w$).
Consider a single query $x$. 
Its rank estimation before and after the compaction defers by at most $w$ regardless of $k$.
This is illustrated in Figure~\ref{compactor}.

\begin{figure}[!ht]
  \centering
      \includegraphics[width=0.8\textwidth]{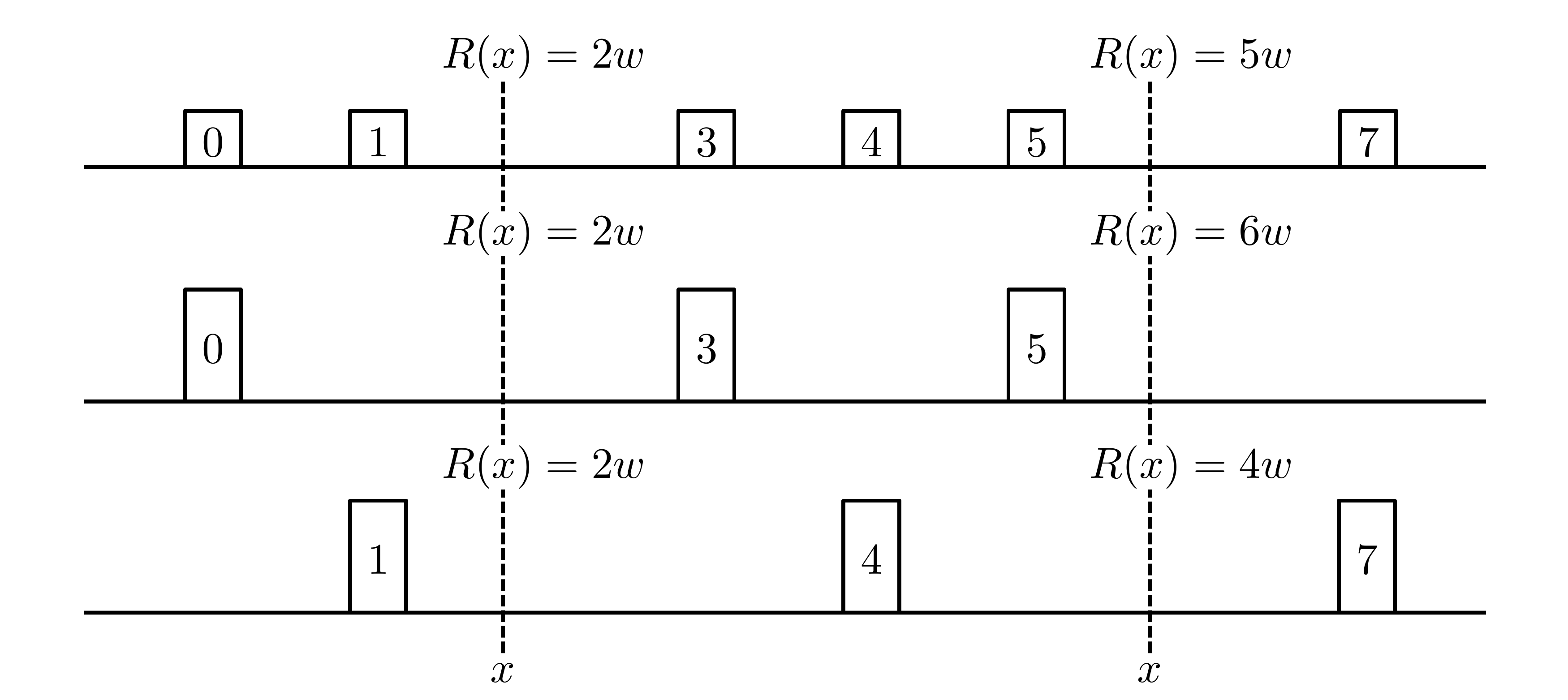}
  \caption{An illustration of a single compactor with $6$ items performing a single compaction operation.
  The rank of a query remains unchanged if its rank with in the compactor is even.
  If it is odd, its rank is increased or decreased by $w$ with equal probability by the compaction operation.}
  \label{compactor}
\end{figure}

This already gives a deterministic algorithm. Assume we use such a compactor. When it outputs items, we feed them into another compactor and so on.
Since each compactor halves the number of items in the sequence, there could be at most $H \le \ceil{\log(n/k)}$ compactors chained together.
Let $h$ denote the height of a compactor where the last one created has height $h=H$ and the first one has height $h=1$.
Let $w_h = 2^{h-1}$ be the weight of items compactor $h$ gets. Then, the number of compact operations it performs is at most $m_h = n/kw_h$.
Summing up all the errors in the system $\sum_{h=1}^{H} m_h w_h = \sum_{h=1}^{H} n/k = Hn/k \le n\log(n/k)/k$.
Since we have $H$ compactors, the space usage is $kH  \le k \log(n/k)$. 
Setting $k = O((1/\eps)\log(\eps n))$ yields error of $\eps n$ with space $O((1/\eps)\log^{2}(\eps n))$.

One conceptual contribution of Agarwal et al.\ \cite{Agarwal12} is to have each compactor delete the odd or even items with equal probability.
This has the benefit that the expected error is zero. It also lets one use standard concentration results to bound the maximal error.
This eliminates one $\log$ factor from the worst case analysis. 
However, the dependence on the failure probability adds a factor of $\sqrt{\log(1/\delta)}$. 
In the all quantiles problem this translates into an additional $\sqrt{\log(1/\eps)}$ factor for constant failure probability.
%
Intuitively Agarwal et al.\ \cite{Agarwal12} also show that when $n \ge \poly(1/\eps)$ one can sample items from the stream before feeding them to the sketch.
This gives total space usage of  $O\left((1/\eps)\log(1/\eps) \sqrt{\log(1/\delta)}\right)$ for the single quantile problem and $O\left((1/\eps)\log(1/\eps) \sqrt{\log(1/\eps\delta)}\right)$ for the all quantiles problem.

The first improvement we provide to the algorithms above is to use different compactor capacities in different heights, denoted by $k_h$.
We show that $k_h$ can, for example, decrease exponentially $k_h \approx k_H (2/3)^{H-h}$.
That is, compactors in lower levels in the hierarchy can operate with significantly less capacity.
Surprisingly enough, this turns out to not effect the asymptotic statistical behavior of the error at all. 
Moreover, the space complexity is clearly improved.

The capacity of any functioning compactor must be at least $2$. This could contribute $O(H) = O(\log(n/k))$ to the space complexity.
To remove this dependence, we notice that a sequence of $H''$ compactors with capacity $2$ essentially perform sampling. 
Out of every $2^{H''}$ elements they select one at random and output that element with weight $2^{H''}$.
This is clearly very efficiently computable and does not truly require memory complexity of $O(H'')$ but rather of $O(1)$.
The total capacity of all compactors whose capacity is more than $2$ is bounded by $\sum_{h= H''+1}^{H} k (2/3)^{H-h} \le 3k$.
This yields a total space complexity of $O(k)$.
Setting $k = O((1/\eps)\sqrt{\log(1/\delta)})$ gives an algorithm with space complexity $O((1/\eps)\sqrt{\log(1/\delta)})$ for the single quantile problem and $O((1/\eps)\sqrt{\log(1/\eps)})$ for the all quantiles problem, which constitutes our first result.
Interestingly, our algorithm can be thought of a smooth interpolation between carful compaction of heavy items and efficient sampling for light items.
We believe this idea would be potentially useful in other geometric coreset construction problems.

%

The next improvement comes from special handling of the top $\log\log (1/\delta)$ compactors.
Intuitively, the number of compaction operations (and therefore random bits) in those levels is so small that one could expect the worst case behavior. 
For worst case analysis a fixed $k_h$ is preferred to diminishing values of $k_h$. 
Therefore, we suggest to set $k_h = k$ when $h \ge H - O(\log\log (1/\delta))$ and $k_h \approx k (2/3)^{H-h}$ otherwise.
By analyzing the worst case error of the top $\log\log (1/\delta)$ compactors separately from the bottom $H - \log\log (1/\delta)$ we improve our analysis to $O((1/\eps)\log^2 \log(1/\delta))$. This sketch is fully mergeable.
Interestingly, the worst case analysis of the top $\log\log (1/\delta)$ compactors is identical to the analysis of the MRL sketch above.

This last observation leads us to our third and final improvement. 
If one replaces the top $\log\log (1/\delta)$ compactors with a GK sketch, the space complexity can be shown to reduce to $O((1/\eps)\log\log(1/\delta))$.
However, this prevents the sketch from being mergeable because the GK sketch is not known to have this property. 

Another way to view this algorithm is as a concatenation of three sketches. The first, receiving the stream of elements is a sampler that simulates all the compactors of capacity $2$. Its output is fed into a sketch composed of a sequence of compactors of increasing sizes, as described above. We refer to such a sketch as a \emph{KLL} sketch. This sketch outputs $O((1/\eps) \poly \log(1/\delta))$ items that are fed into an instance of GK.
This idea is illustrated in Figure~\ref{three}.

\hspace{3cm}
\begin{figure}[!ht]
\hspace{1cm}
\begin{minipage}{1\textwidth}
      \includegraphics[width=0.8\textwidth]{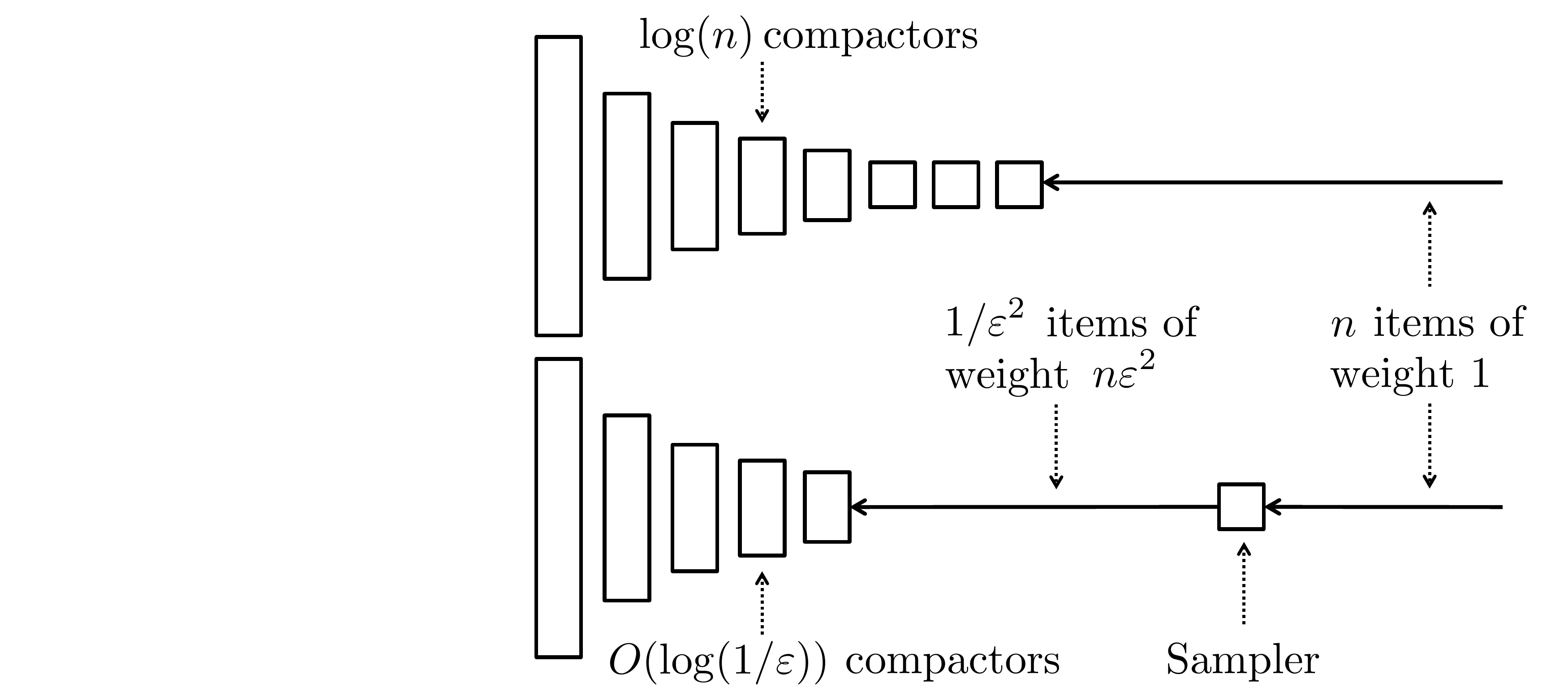}\\
      \includegraphics[width=0.8\textwidth]{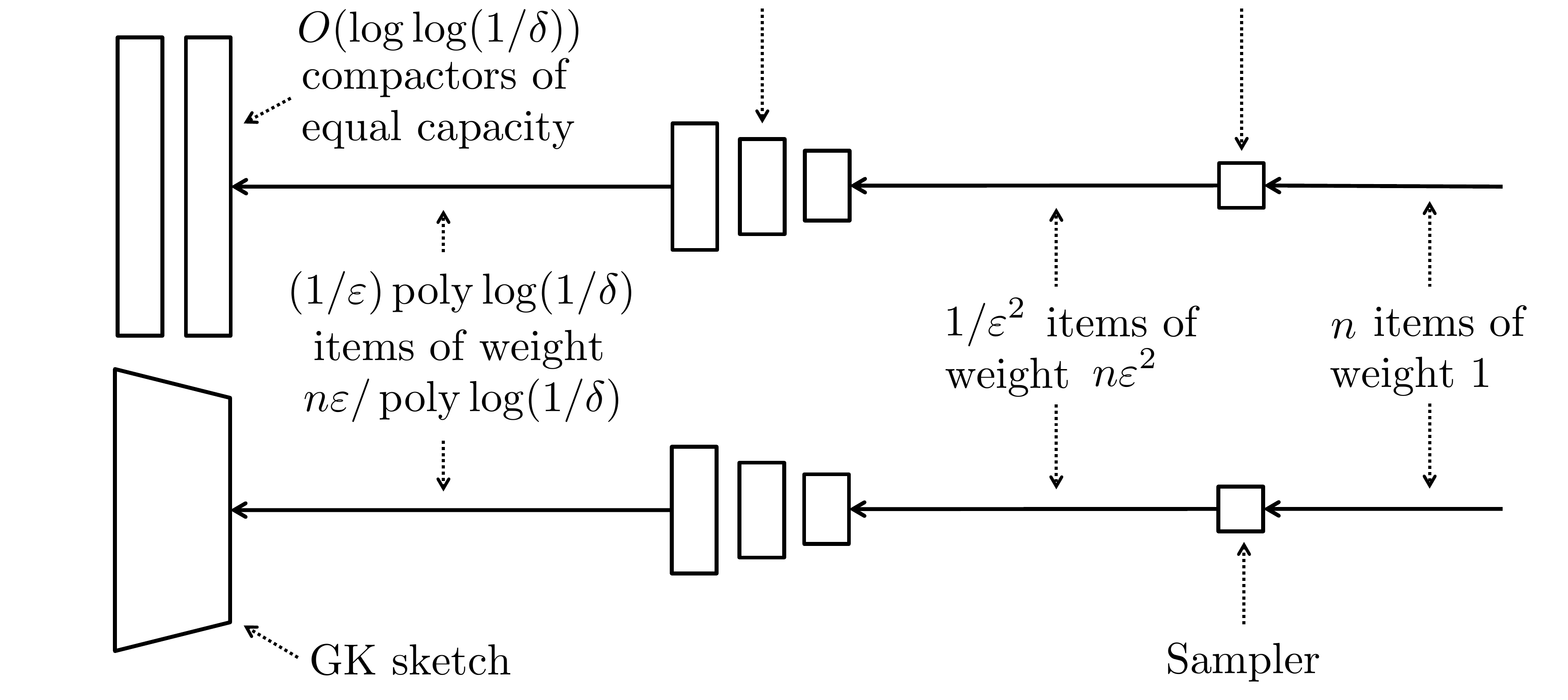}
\end{minipage}
\caption{Sampling, varying capacity compactors, equal capacity compactors and GK sketches achieve different efficiencies at different stream lengths.
The set of capacitated compactors used by KLL creates a smooth transition between sampling and GK sketching.
The top figure corresponds to the first construction explained in Section~\ref{analysis}. The second from the top is explained in Section~\ref{reducing}.
The two bottom figures correspond to our main contributions and are explained in Sections~\ref{reducing1} and~\ref{reducing2} respectively.}

\label{three}
\end{figure}

\newcommand{\tp}[2]{\parbox{#1cm}{\vspace{0.1cm}#2\vspace{0.1cm}}}

\begin{table}[h!]
{
\footnotesize
\begin{center}
\begin{tabular}{|l|c|c|c|c|c|} \hline
	                         				& Single quantile  			& All quantiles	& \tp{1.9}{Randomized} &  \tp{1.7}{Mergeable}  \\ \hline
\tp{2}{MRL \cite{Manku99}}		& $(1/\eps)\log^{2}(\eps n)$	& $(1/\eps)\log^{2}(\eps n)$ 	& No 		& Yes 		 \\ \hline
\tp{2}{GK \cite{GreenwaldK01}}		& $(1/\eps)\log(\eps n)$		& $(1/\eps)\log(\eps n)$ 		& No 		& No 				 \\ \hline
\tp{2}{ACHPWY \cite{Agarwal12}}	& \tp{3.8}{$(1/\eps)\log(1/\eps)\sqrt{\log(1/\delta)}$}	& \tp{3.8}{$(1/\eps)\log(1/\eps) \sqrt{\log(1/\eps\delta)}$} 	& Yes 		& Yes \\ \hline  
\tp{2}{FO \cite{FelberO15a}\\${\small \delta=e^{-(1/\eps)^{c}}}$}		& $(1/\eps)\log(1/\eps)$		& $(1/\eps)\log(1/\eps)$ & Yes & No \\ \hline
\tp{2}{KLL \newline [This paper] }	& $(1/\eps) \log^2 \log(1/\delta)$				& $(1/\eps)\log^2\log(1/\delta\eps)$ 	& Yes		& Yes  \\ \hline
\tp{2}{KLL  \newline [This paper]} & $(1/\eps)\log\log(1/\delta)$				& $(1/\eps)\log\log(1/\delta\eps)$	& Yes		& No \\ \hline
\end{tabular}
\end{center}
}
\caption{The table describes the space complexity of several streaming quantiles  algorithms in big-$O$ notation.  
The randomized algorithms are required to succeed with constant probability.  They all work in the comparison model and for arbitrarily ordered streams.
The non-mergeability of some of these algorithms is due to the fact that they use GK, which is not known to be mergeable, as a subroutine. }
\label{contib}
\end{table}%

\section{Algorithm and Analysis}\label{analysis}
As mentioned above, our sketching algorithm (KLL) includes a hierarchy of compactors with varying capacities. 
Consider a run of the algorithm that terminates with $H$ different compactors.
The compactors are indexed by their hight $h \in 1,\ldots,H$. 
The weight of items at hight $h$ is $w_h=2^{h-1}$.
Denote by $k_h$ the smallest number of items that the compactor at height $h$ contains during a compact operation.
For brevity, denote $k = k_H$.
For reasons that will become clear later assume that $k_h \ge k c^{H-h}$ for $c \in (0.5,1)$.

Since the top compactor was created, we know that the second compactor from the top compacted its elements at least once. 
Therefore $n \ge k_{H-1}w_{H-1} =  k_{H-1}2^{H-2}$ which gives
\[
H \le \log(n/k_{H-1})+2 \le \log(n/ck)+2 \ .
\]
Using the above we can bound the number of compact operations $m_h$ at height $h$.
Every compact procedure call is performed on at least $k_h$ items and the items have a weight of $w_h = 2^{h-1}$. Therefore, 
$$m_h \le \frac{n}{k_h w_h} \le \frac{2n}{k2^H} (2/c)^{H-h} \le (2/c)^{H-h-1} \ .$$

To analyze the total error produced by the sketch, we first consider the error generated in each individual level.
Define by $\rank(k,h)$ the rank of $x$ among the following weighted set of items. 
The items yielded by the compactor at height $h$ and all the items stored in the compactors of heights $h' \le h$ at end of the stream.
For convenience $\rank(x,0) = \rank(x)$ is the exact rank of $x$ in the input stream.
Define $\var{err}(x,h) = \rank(x,h) - \rank(x,h-1)$ to be the total change in the approximated rank of $x$ due to level $h$. 

Note that each compaction operation in level $h$ either leaves the rank of $x$ unchanged or adds $w_h$ or subtracts $w_h$ with equal probability.
To be more explicit, if $x$ has an even rank among the item inside the compactor, the total mass to the left of it (its rank) is unchanged by the compaction operation.
If its rank inside the compactor is odd however and the odd items are chosen, this mass increases by $w_h$.
If its rank inside the compactor is odd and the even items are chosen, its mass decreases by $w_h$. 
Therefore, $\var{err}(x,h) = \sum_{i=1}^{m_h} w_h X_{i,h}$ where $\E[X_{i,h}] = 0$ and $|X_{i,h}| \le 1$.
The final discrepancy between the real rank of $x$ and its approximation $\tilde{\rank}(x) = \rank(x,H)$ is
\begin{equation*} 
\rank(x,H) - \rank(x,0) =  \sum_{h=1}^{H}  \rank(x,h) - \rank(x,h-1) = \sum_{h=1}^{H} \var{err}(x,h)  = \sum_{h=1}^{H} \sum_{i=1}^{m_h} w_h X_{i,h} \ .
\end{equation*}
\begin{lemma} [Hoeffding] \label{lem:hoeffding}
Let $X_1,\ldots,X_m$ be independent random variables, each with an expected value of zero, taking values in the range $[-w_i,w_i]$. Then for any $t>0$ we have
$$ \Pr\left[ \left| \sum_{i=1}^m X_i \right|>t \right] \le 2\exp\left( -\frac{t^2}{2 \sum_{i=1}^m w_i^2} \right) $$
with $\exp$ being the natural exponent function.
\end{lemma}
We now apply Hoeffding's inequality to bound the probability that the bottom $H' \leq H$ compactors contribute more than $\eps n$ to the total error. 
The reason for considering only the bottom levels and not all the levels will become apparent in Section~\ref{reducing}.
\begin{equation}\label{epsn}
\Pr\left[ \left| \rank\left(x,H' \right) - \rank(x,0) \right| > \eps n \right] =  \Pr\left[ \sum_{h=1}^{H'} \sum_{i=1}^{m_h} w_h X_{i,h} >  \eps n \right] \le 2\exp\left( -\frac{\eps^2 n^2}{2 \sum_{h=1}^{H'} \sum_{i=1}^{m_h} w^2_h} \right)\end{equation}
A straight forward computation shows that 
\begin{eqnarray*}
\sum_{h=1}^{H'} \sum_{i=1}^{m_h} w^2_h &=& \sum_{h=1}^{H'} m_h w^2_h \le  \sum_{h=1}^{H'} (c/2)^{H' - h -1} 2^{2h-1} \\
&=& \frac{(2/c)^{H'-1}}{4} \sum_{h=1}^{H'}(2c)^h \le \frac{(2/c)^{H'-1}}{4} \frac{(2c)^{H'}}{2c-1} \le \frac{c}{8(2c-1)}2^{2H'}
\end{eqnarray*}

\noindent Substituting $2^{2H'} = 2^{2H}/2^{2(H-H')}$ and recalling that $H \le \log(n/ck) +2$ we get that
\begin{equation}\label{sumw}
\sum_{h=1}^{H'} \sum_{i=1}^{m_h} w^2_h \le \frac{n^2/k^2}{2c^2(2c-1)} \frac{1}{2^{2(H-H')}}
\end{equation}

\noindent Substituting Equation~\ref{sumw} into Equation~\ref{epsn} and setting $C = c^2(2c-1)$ we get the following convenient form
\begin{equation}\label{nice1}
\Pr\left[ |\rank(x,H') - \rank(x)| \ge \eps n\right] \le 2 \exp \left( -C\eps^2 k^2 2^{2(H-H')} \right)
\end{equation}

\begin{theorem}\label{kll0}
There exists a streaming algorithm that computes an $\eps$ approximation for the rank of a single item with probability $1-\delta$ whose space complexity is $O((1/\eps)\sqrt{\log(1/\delta)} + \log(\eps n))$. This algorithm also produces mergeable summaries.
\end{theorem}
\begin{proof}
Let $k_h = \ceil{k c^{H-h}}+1$. Note that $k_h$ changes throughout the run off the algorithm. Nevertheless, $H$ can only increase and so $k_h$ is monotonically decreasing with the length of the stream. 
This matches the requirement that $k_h \ge k c^{H-h}$ where $H$ is the final number of compactors.
Notice that $k_h$ is at least $2$.

Setting $H'=H$ in Equation~\ref{nice1} and requiring failure probability at most $\delta$ we conclude that it suffices to set $k = (C/\eps)\sqrt{\log(2/\delta)}$.
The space complexity of the algorithm is $O\left( \sum_{h=1}^{H} k_h \right)$.
\begin{eqnarray*}
\sum_{h=1}^{H} k_h \le \sum_{h=1}^{H} (k c^{H-h} + 2) \le k/(1-c) + 2H = O(k + \log(n/k))
\end{eqnarray*}
Setting $\delta = \Omega(\eps)$ suffices to union bound over the failure probabilities of $O(1/\eps)$ different quantiles.
This provides a mergeable sketching algorithm for the all quantiles problem of space $O((1/\eps)\sqrt{\log(1/\eps)} + \log(\eps n))$.

The KLL sketch provides a \emph{mergeable summary}.
In a merge operation, same height compactors are concatenated together.
Then, each level that contains more than $k_h$ elements is compacted. 
The value of $k_h$ is based on the new maximal height $H$ which is derived from the combined lengths of the two streams.
In either of the two merged sketches, each compaction at level $h$ involved at least $k_h$ items which means the proof above still holds.

\end{proof}
\noindent Note that this result already improves on the best known prior art in the parameter setting where $n = \exp(O(1/\eps))$.

\begin{theorem}\label{kll0}
There exists a streaming algorithm that computes an $\eps$ approximation for the rank of a single item with probability $1-\delta$ whose space complexity is $O\left((1/\eps)\sqrt{\log(1/\delta)} \right)$.
\end{theorem}
\begin{proof}
The KLL sketch maintains a total of $H = \log(\eps n)$ sketches. 
Note, however, that only $O(\log(k))$ compactors have capacity greater than $2$.
More accurately the bottom $H'' =H - \ceil{\log(k)/\log(1/c)}$ all have capacity exactly $2$.
Each of those receives two items at a time, performs a random match between them, and sends the winner of the match to the compactor of the next level. Hence, the compactor of level $H''$ simply selects one item uniformly at random from every $2^{H''}$ elements in the stream and passes that item with weight $2^{H''}$ to the compactor at hight $H''+1$. 
This is easily simulated using $O(1)$ space which replaces the bottom $H''$ compactors and reduces the space complexity of the algorithm.
\end{proof}

There is, however, a drawback in replacing the bottom $H''$ compactors with a simple sampler.
When merging two sketches, it is not clear how to merge the samplers in a correct way.
The next section explains how to do exactly that.

\section{Sampling and Keeping Mergeablity}\label{reducing}



In the new sketch we have, in addition to the compactors, a new object we call a sampler. 
The sampler supports an $\func{update}$ method that introduces an item of weight $w$ to the sketch. When observing items in a stream the weight is always set to $1$. However, when merging two sketches we require supporting an update of arbitrary weights.
%
At any time the sampler has an associated height $h$. 
It outputs items of weight $2^h$ as inputs to the compactor of level $h+1$.
This associated height will increase over time to eventually being roughly $H-\log(1/\eps)$.
%
%
Apart for sampler of height $h$, the sketch maintains compactors at heights greater than $h$.

The sampler keeps a single item in storage along with a weight of at most $2^h-1$. When merging two sketches, the sketch with the sampler of smaller height will feed its item with its appropriate weight to the sampler of the other sketch. Also, all compactors with height $\leq h$ in the `smaller' sketch will feed the items in their buffers to the sampler of the `larger sketch'.

The $\func{update}$ operation is performed as follows. Denote by $v$ the weight of the internal item stored in the sampler and by $w$ the weight of the newly introduced item. If $v+w \leq 2^{h}$, the sampler replaces its stored item with the new one with probability $w/(v+w)$ as in Reservoir Sampling. 
If $v+w = 2^h$ the sampler outputs the stored item and sets the internal weight $w$ to $0$. If however $v+w > 2^h$ (notice this can only happen if $w>1$) the sampler discards the heavier item, and keeps the lighter item with a weight of $\min\{w,v\}$. With probability $\max\{w,v\}/2^h$ it also outputs the heavier item with weight $2^h$. 

It is easy to verify that the above described sketching scheme corresponds to the following offline operations. 
The sampler outputs items by performing the action \emph{sample}; 
it takes as input a sequence of items $W$ items such that $2^{h-1} < W \leq 2^h$.
With probability $W/2^h$ it outputs one of the observed items chosen at random.
With probability $1-W/2^h$ it outputs nothing. 
Before analyzing the error associated with a \emph{sample} operation we mention that, if no merges are performed, it suffices to restrict the value of $W$ to be exactly $2^h$. 
However, in order to account for merges we must let $W$ obtain values in the range $W \in (2^{h-1},  2^h]$.

\begin{lemma}
Let $\rank(x,h,i)$ be the rank of $x$ after sample operation $i$ of the sampler of height $h$, where the rank is computed based on the stream elements that did not undergo a sample operation, and the weighted items outputted by the sample operations $1$ through $i$.
Let $\rank(x,h,i) - \rank(x,h,i-1) = 2^h Y_{h,i}$. Then $\E[Y_{h,i}] = 0$ and $|Y_{h, i}| \le 1$.
\end{lemma}
\begin{proof}
Let $r$ be the exact rank of $x$ among the input items before the sampling operation. Denote by $W$ the sum of weights of the input items.
After the sampling, those input items are replaced with a single element. 
The rank of $x$ after the sampling is $2^h$ with probability $\frac{W}{2^h}\cdot \frac{r}{W}$ and $0$ otherwise. 
The first term is the probability of outputting anything. The second is of selecting an element smaller than $x$.
Therefore, the expected value of the rank of $x$ after the sample operation is $2^h\frac{W}{2^h}\cdot \frac{r}{W}  = r$, hence the expected value of the difference mentioned in the claim is $0$.
Clearly, the maximal value of this difference is bounded by $2^h$.
\end{proof}

Let $m_h$ be total number of times a \emph{sample} operation can be performed at height $h$.
Since the sampler at that height takes items with a total minimum weight of $W > 2^{h-1}$ and there are a total of $n$ items (with overall weight $n$)
$$ m_h \leq \frac{n}{2^{h-1}} \ .$$
It follows that the expression of the error accounted for samplers of heights up to\footnote{Notice that unlike compactors, there is no hierarchy of samplers. However, due to the fact that the height of the sampler grows with time, the sample operations may be performed on different heights. The only guarantee is that any item appears in at most one sample operation and that the height of the sampler is always at most $H''$.} $H''$ can be expressed as 
\begin{equation*}
\var{err}_{H''} =  \sum_{h=1}^{H''} \sum_{i=1}^{m_h}[\rank(x,h,i) - \rank(x,h,i-1)] = \sum_{h=1}^{H''} \sum_{i=1}^{m_h} 2^{h} Y_{i,h}
\end{equation*} 
with $Y_{i,h}$ being independent, $\E[Y_{i,h}]=0$ and $|Y_{i,h}| \le 1$. We compute the sum of weights appearing in Hoeffding's inequality (Lemma~\ref{lem:hoeffding}) in order to apply it.
$$ \sum_{h=1}^{H''} \sum_{i=1}^{m_h} 2^{2h} \leq \sum_{h=1}^{H''} n 2^{h+1} \leq 4n 2^{H''} = \frac{4 n 2^H}{2^{H-H''} } \leq \frac{16 n^2 }{ck 2^{H-H''} }$$
Hence,
\begin{equation} \label{eq:nice_sampler}
\Pr\left[ \var{err}_{H''} > \eps n \right] \leq 2\exp\left(- c \eps^2 k 2^{H-H''}/32 \right) \ .
\end{equation}

The following is immediate from Equations~\ref{nice1} and~\ref{eq:nice_sampler}.
\begin{theorem} \label{thm:space_saved}
Assume we apply a KLL sketch that uses $H$ levels of compactors, with capacity $k_h \geq k c^{H-h}$. Also, assume that an arbitrary subset of the stream is fed into samplers of heights $1$ through $H''$, while the output of these samplers is fed to appropriate compactors.  Then for any $H' > H''$ it holds that
$$\Pr\left[ \var{err}_{H'} > 2\eps n \right] < 2\exp\left(- c \eps^2 k 2^{H-H''}/32 \right) + 2\exp \left( -C\eps^2 k^2 2^{2(H-H')} \right)$$
Here, $\var{err}_{H'}$ denotes the error of the stream outputted by the compactors of level $H'$, and $C = c^2(2c-1)$.
\end{theorem}

By taking $H'=H$, $H''=H-O(\log(k))$, and $k=(1/\eps)\sqrt{\log(1/\delta)}$ we obtain the following corollary.

\begin{corollary}\label{kll1}
There exists a streaming algorithm that computes an $\eps$ approximation for the rank of a single item with probability $1-\delta$ whose space complexity is $O((1/\eps)\sqrt{\log(1/\delta)})$. This algorithm also produced mergeable summaries.
\end{corollary}
%

\section{Reducing the Failure Probability}\label{reducing1}

In this section we take full advantage of Theorem~\ref{thm:space_saved} to obtain a streaming algorithm with asymptotically better space complexity. 
Notice that the lion share of the contribution to the error is due to the top compactors.
For those, however, Hoeffding's bound is not tight. Let $s=O(\log\log(1/\delta))$ be a small number of top layers. 
For the bottom $H-s$ layers we use Theorem~\ref{thm:space_saved}, applied on $H'=H-s$ and corresponding $H''$, to bound their error. For the top $s$ we simply use a deterministic bound.

\begin{theorem}\label{kll2}
There exists a streaming algorithm that computes an $\eps$ approximation for the rank of a single item with probability $1-\delta$ whose space complexity is $O((1/\eps)\log^2\log(1/\delta))$. This algorithm also produces mergeable summaries.
\end{theorem}

\begin{proof}
Using Theorem~\ref{thm:space_saved} we see that the bottom compactors of height at most $H'=H-s$ and the sampler, when set to be of height at most $H''=H-2s-\log_2(k)$, contribute at 
most $\eps n$ to the error with probability $1-\delta$ at long as $\eps k 2^s \ge c' \sqrt{\log(2/\delta)}$ for sufficiently small constant $c'$.
For the top $s$ compactors, we set their capacities to $k_h = k$. 
That is, we do not let their capacity drop exponentially.
Those levels contribute to the error at most $\sum_{h=H'+1}^{H} m_h w_h = \sum_{h=H'+1}^{H} n/k_h = sn/k$.
Requiring that this contribution is at most $\eps n$ as well we obtain the relation $s \le k \eps$.
Setting $s = O(\log \log (1/\delta))$ and $k = O(\frac{1}{\eps}\log \log (1/\delta))$ satisfies both conditions.
The space complexity of this algorithm is dominated by maintaining the top $s$ levels which is $O(ks) = O((1/\eps)\log^2\log(1/\delta))$. 
\end{proof}

Interestingly, the analysis of the top $s$ levels is identical to the equal capacity compactors used in the MRL sketch. 
In the next section we show that one could replace the top $s$ levels with a different algorithm and reduce the dependence on $\delta$ even further. 

\subsection{Gaining Space Optimality; Potentially Losing Mergeability}\label{reducing2}
The most space efficient version of our algorithm, with respect to the failure probability, operates as follows.
For $\delta$ being the target error probability we set 
$$s = \ceil{\log_2\left( c' \sqrt{\log(2/\delta)} / (k\eps) \right)} = O(\log \log (1/\delta))$$ 
as in the section above but set $k = O(1/\eps)$.
At any time point we keep 2 different copies of the GK sketch, tuned for a relative error of $\eps$. They are correspondingly associated with the compactors of heights $h_1<h_2$ which are the two largest height values that are multiples of $s$. The GK sketch associated with height $h$ receives as input the outputs of the compactor of layer $h-1$. For $h=0$ the GK sketch associated with it receives as input the stream elements. When a new GK sketch is built due to a new compactor being formed  the bottom one is discarded.

\begin{theorem}\label{kll3}
There exists a streaming algorithm that computes an $\eps$ approximation for the rank of a single item with probability $1-\delta$ whose space complexity is $O((1/\eps)\log\log(1/\delta))$.
\end{theorem}
\begin{proof}
Notice that the height of $h_1$ is at least $H-2s$. 
It follows that the total number of items that is ever fed into a single GK sketch at most $n_{1} = k 2^{2s}$.
Applying Theorem~\ref{thm:space_saved}  again, on $H' =H-s$, and $H''=H-2s-\log_2(k)$, the error w.r.t.\ to the output of the compactor feeding elements into the GK sketch matching $h_1$ is at most $O(\eps n)$, with $k=O(1/\eps)$. 
Therefore, the sum of errors is still $O(\eps n)$. The memory required by the GK sketch with respect to its input is at most 
$O((1/\eps) \log (\eps n_{1})) = O\left( (1/\eps) s \right) = O\left( (1/\eps) \log \log(1/\delta) \right)$, which dominates the memory of our sketch with $k=O(1/\eps)$. The claim follows.
\end{proof}

We note that the $GK$ sketch is not known to be fully mergeable and so that property of the sketch is lost by this construction. That being said, we point out that the $GK$ sketch is \emph{one-way mergeable}. One-way mergeability is a weaker form of mergeability that informally states that the following setting can work: The data is partitioned among several machine, each creates a summary of its own data, and a single process merges all of the summaries into a single one. For example, stated in terms of Definition~\ref{def:merge}, when the error $\eps(S(N),N)$ is linear in the size of the sketch $N$, i.e., $\eps(S(N),N) = \eps_S |N|$,
a merge operation resulting in an error of $\eps_S |N_1| + 2\eps_S |N_2|$ rather than $\eps_S |N_1| + \eps_S |N_2|$ is one-way mergeable, while it is not (fully) mergeable. It was pointed out by~\cite{GK2016,Agarwal12} that any sketch for the all quantiles approximation problem is one-way mergeable.

\section{Tightness of Our Result}

In \cite{hung2010omega} a lower bound is given for deterministic algorithms. A more precise statement of their result, implicitly shown in the proof of their Theorem 2, is as follows
\begin{lemma} [Implicit in \cite{hung2010omega}, Theorem 2] \label{lem:lb_specific}
Let $\cal A_D$ be deterministic comparison based algorithm solving single quantile $\eps$ approximate for all streams of length at most $C (1/\eps)^2 \log(1/\eps)^2$ for some sufficiently large universal constant $C$. 
Then $\cal A_D$ must store at least $c (1/\eps) \log(1/\eps)$ elements from the stream for some sufficiently small constant $c$. 
\end{lemma}
Below we obtain a lower bound matching our result completely for the case of single quantile problem and almost completely for the case of $\eps$ approximation of all quantiles.
\begin{theorem}
Let $\cal A_R$ be a randomized comparison based algorithm solving the $\eps$ approximate single quantile problem with probability at least $1-\delta$. 
Then $\cal A_R$ must store at least $\Omega\left( (1/\eps) \log \log(1/\delta) \right)$ elements from the stream.
\end{theorem}
\begin{proof}
Assume by contradiction that there exists a randomized algorithm $\cal A_R$ that succeeds in computing a single quantile approximation up to error $\eps$ with probability $1-\delta$ while storing $o\left( (1/\eps) \log \log(1/\delta) \right)$ elements from the stream.
Let $n$ be the length of the stream and $\delta = 1/2n!$. 
With probability $1/2$ the randomized algorithm succeeds simultaneously for all $n!$ possible inputs.
Let $r$ denote a sequence of random bits used by $\cal A_R$ in one of these instances. 
It is now possible to construct a deterministic algorithm ${\cal A}_D(r)$ which is identical to ${\cal A}_R$ but with $r$ hardcoded into it.
Note that ${\cal A}_D(r)$ deterministically succeeds for streams of length $n$. 
Let $n=C (1/\eps)^2 \log(1/\eps)^2$ for the same $C$ as in Lemma~\ref{lem:lb_specific}.
We obtain that ${\cal A}_D(r)$ succeeds on all streams of length $C (1/\eps)^2 \log(1/\eps)^2$ while storing $o((1/\eps) \log(1/\eps))$ elements from the stream.
This contradicts Lemma~\ref{lem:lb_specific} above.
\end{proof}
\section{Discussion}
The lower bound above perfectly matched the single quantile approximation result we achieve.
For the all quantiles problem, using the union bound over a set of $O(1/\eps)$ quantiles shows that $O((1/\eps)\log\log(1/\eps))$ elements suffice. 
This leaves a potential gap of $\log\log(1/\eps)$ for that problem.

\section*{Acknowledgments}
The authors want to thank Sanjeev Khanna, Rafail Ostrovsky, Jeff Phillips, Graham Cormode, Andrew McGregor and Piotr Indyk for very helpful discussions and pointers.

\end{document}